\documentclass[11pt,a4paper]{article}

\usepackage[utf8]{inputenc}
\usepackage[T1]{fontenc}
\usepackage{lmodern}
\usepackage{amsmath,amssymb,amsthm}
\usepackage{algorithm}
\usepackage{algpseudocode}
\usepackage{array}
\usepackage{booktabs}
\usepackage{graphicx}
\usepackage{enumitem}
\usepackage{xspace}
\usepackage{tikz}
\usetikzlibrary{arrows.meta,positioning,shapes.geometric,calc,fit,decorations.pathmorphing}
\usepackage[margin=1in]{geometry}
\usepackage{caption}
\usepackage{subcaption}
\usepackage{hyperref}

\usepackage{xcolor}

\newtheorem{theorem}{Theorem}[section]
\newtheorem{lemma}[theorem]{Lemma}

\newtheorem{proposition}[theorem]{Proposition}
\theoremstyle{definition}
\newtheorem{definition}{Definition}[section]
\newtheorem{assumption}{Assumption}[section]
\theoremstyle{remark}
\newtheorem{remark}{Remark}[section]

\newcommand{\REV}{\mathsf{REV}}
\newcommand{\SA}{\mathsf{SA}}
\newcommand{\AP}{\mathcal{AP}}
\newcommand{\Ctx}{\mathsf{Ctx}}
\newcommand{\KeyAP}[1][]{\mathsf{Key}_{\AP_{#1}}}
\newcommand{\Cred}{\mathsf{Cred}}
\newcommand{\AF}{\mathsf{AdminFactor}}
\newcommand{\UC}{\mathsf{UserCred}}
\newcommand{\Adv}{\mathcal{A}}
\newcommand{\negl}{\mathsf{negl}}
\newcommand{\secparam}{\lambda}
\newcommand{\getsr}{\stackrel{\$}{\leftarrow}}
\newcommand{\concat}{\,\|\,}
\newcommand{\ACEGF}{\textsf{ACE-GF}\xspace}
\newcommand{\CTDAP}{\textsf{CT-DAP}\xspace}

\title{Condition-Triggered Cryptographic Asset Control\\via Dormant Authorization Paths}

\author{
  \texttt{Jian Sheng Wang} \\
  \texttt{jason@acechain.io} \\
  Montreal, Canada
}
\date{November 24, 2026}

\begin{document}
\maketitle

\begin{abstract}
Control of encrypted digital assets is traditionally equated with permanent
possession of private keys.  This model leaves no room for regulatory
supervision, conditional delegation, or legally compliant transfer at the
cryptographic layer.  Existing remedies, including multi-signature schemes,
threshold signature protocols, smart-contract-based enforcement, and custodial
delegation, require either persistent key exposure, on-chain state mutation,
or trusted intermediaries, and therefore fail to provide immediate, revocable,
and condition-dependent control at the cryptographic layer.

We introduce \emph{Condition-Triggered Dormant Authorization Paths}
(\CTDAP), a cryptographic asset control method built on
\emph{destructible authorization factors}.  The method is parameterized by
an abstract \emph{root-derivable framework} satisfying three properties:
deterministic key derivation from a single root entity, context-isolated
control capability generation, and authorization-bound revocation.
Under \CTDAP, asset control rights are modeled as \emph{dormant
authorization paths} composed of user-held credentials and one or more
administrative authorization factors held by independent custodians.
A control right remains cryptographically inactive until all required
authorization factors are simultaneously available.

Upon verification of predefined external conditions (such as explicit user
consent, legally recognized inheritance events, or time-based
triggers) by a designated verifier such as a notary or oracle, the corresponding administrative authorization factor is
released, activating the associated control path.  Revocation is achieved
by destroying authorization factors, which prevents any further
activation of the corresponding control path without altering the
underlying cryptographic root.  A party that has already activated a path
holds the root; removing its control requires moving the assets to a
fresh root.

We formalize the threat model, define security games for unauthorized
control resistance, path isolation, and stateless revocation, and provide
proofs under standard assumptions (IND-CPA and INT-CTXT security of
AES-GCM-SIV as an AEAD, PRF security of HKDF in the random-oracle
model, memory-hardness of Argon2id, and collision resistance of
SHA-256).  We
instantiate the framework using the Atomic Cryptographic Entity
Generative Framework (ACE-GF) and evaluate performance on commodity
hardware, demonstrating sub-second activation latency with
configurable security--performance trade-offs.

\end{abstract}

\section{Introduction}\label{sec:intro}

\subsection{Background and Motivation}\label{sec:background}

In prevailing designs, control rights over encrypted digital assets are
implicitly equated with permanent possession of cryptographic secrets,
typically private keys.  This model is simple and gives the holder strong
autonomy, but it cannot express the asset control needed in
real-world scenarios that require conditional delegation, regulatory
supervision, inheritance, or legally enforceable transfer of rights.

In particular, existing cryptographic asset systems lack native mechanisms
for representing control rights that are \emph{temporarily inactive},
\emph{conditionally enabled}, or \emph{revocable without state migration}.
As a result, control transitions are commonly implemented through ad hoc
constructions (trusted custodians, contractual agreements
external to the cryptographic system, or on-chain smart contracts) that
introduce additional attack surfaces and operational complexity.

These limitations are increasingly problematic as encrypted assets are
integrated into regulated financial systems, estate planning, and
multi-party organizational governance, where control must be conditional
and legally compliant as well as secure.

\subsection{Limitations of Existing Approaches}\label{sec:limitations}

Several classes of solutions have been proposed to address aspects of
conditional or shared control; each suffers from structural limitations.

\paragraph{Permanent Key Possession.}
Systems that rely on a single controlling private key provide no intrinsic
mechanism for conditional access or delegation.  Any transfer of control
requires direct disclosure or replacement of the controlling secret,
resulting in irreversible loss of exclusivity and making revocation
impractical~\cite{bip32,bip39}.

\paragraph{Multi-Signature and Threshold Schemes.}
Multi-signature constructions~\cite{goldfeder2015multisig} and threshold
signature protocols (e.g., GG18~\cite{gg18}, GG20~\cite{gg20},
FROST~\cite{frost}) distribute control across multiple keys and parties
but require persistent availability of all participating keys.  Changes in
control conditions necessitate reconfiguration of signing policies or
migration of assets, introducing operational overhead.  Moreover, these
schemes conflate control enforcement with participation requirements,
making revocation and conditional delegation stateful and
inflexible~\cite{gennaro1996robust}.

\paragraph{Secret Sharing.}
Shamir's secret sharing~\cite{shamir79} and its proactive
variants~\cite{proactivess} can distribute a master secret among
multiple parties with a threshold reconstruction policy.  However,
reconstructing the secret reveals it to the reconstructing party,
and the scheme does not natively support conditional activation
tied to external events or stateless revocation of individual shares.

\paragraph{Smart-Contract-Based Enforcement.}
Smart contracts enable programmable control logic, including conditional
execution, but at the cost of on-chain state mutation, code immutability
risks, and dependency on external data
sources~\cite{wood2014ethereum,szabo1997smart}.  Conditional triggers
implemented via contracts are vulnerable to oracle
manipulation~\cite{adler2018astraea}, logic errors, and jurisdictional
ambiguity, and often fail to align cryptographic enforcement with
off-chain legal authorization.

\paragraph{Attribute-Based and Policy-Based Encryption.}
Ciphertext-policy attribute-based encryption
(CP-ABE)~\cite{bethencourt2007cpabe} and key-policy ABE
(KP-ABE)~\cite{goyal2006kpabe} enable fine-grained access control
policies embedded in ciphertexts or keys.  While these schemes support
conditional decryption, they are designed for data confidentiality
rather than cryptographic asset control, and do not address revocable
control paths or stateless revocation without re-encryption.

\paragraph{Time-Lock Cryptography.}
Time-lock puzzles~\cite{rivest1996timelock} and verifiable delay
functions (VDFs)~\cite{boneh2018vdf} enable time-conditioned
release of secrets.  However, they support only temporal conditions
(not arbitrary external events), and revocation after puzzle generation
is impossible.

\paragraph{Custodial Models.}
Custodial solutions externalize control to trusted intermediaries,
sacrificing cryptographic autonomy and introducing centralized trust
assumptions.  While custodians may enforce legal or regulatory
conditions, cryptographic control is no longer directly tied to the
asset holder or beneficiaries~\cite{he2020smart}.

\paragraph{Social Recovery.}
Social recovery wallets~\cite{wohrer2021social} allow
designated guardians to assist in key recovery, but conflate recovery
with ongoing control and do not support dormant, conditionally activated
control paths.

\medskip
Each of these approaches ties \emph{control} to \emph{key possession},
\emph{contract state}, or \emph{institutional trust}, and none offers a
cryptographic abstraction for conditional and revocable asset control.

\subsection{Problem Statement}\label{sec:problem}

We address the following problem:

\begin{quote}
\itshape
How can control of encrypted digital assets be conditionally enabled,
revoked, and transferred in a manner that is cryptographically
enforceable, legally auditable, and does not require persistent key
exposure, asset migration, or on-chain state modification?
\end{quote}

A satisfactory solution must satisfy several requirements simultaneously:

\begin{itemize}[leftmargin=2em]
  \item \textbf{Cryptographic Enforceability:} Unauthorized parties must
    be unable to exercise control regardless of off-chain agreements.
  \item \textbf{Conditional Activation:} Control rights should remain
    inactive until predefined external conditions are satisfied.
  \item \textbf{Immediate Revocation:} While no path over the same
    root has been activated, control must be disableable without
    rotating keys or modifying asset ownership structures.
  \item \textbf{Statelessness:} Control transitions should not require
    maintaining mutable global state.
  \item \textbf{Legal Compatibility:} Cryptographic control should align
    with legally recognized authorization and delegation processes.
\end{itemize}

Existing systems satisfy at most a subset of these requirements, and often
only by introducing additional trusted components.

\subsection{Core Insight}\label{sec:insight}

Our starting point is the following observation:

\begin{quote}
\itshape Control of encrypted assets need not be transferred; it can be
activated.
\end{quote}

We model control rights as \emph{authorization paths} that exist prior to
activation and remain cryptographically dormant until all required
authorization factors are present.  Activation enables a pre-established
authorization path under verifiable conditions; no key is reassigned and
no asset is modified.  Revocation destroys or invalidates an
authorization factor, which prevents further use of the corresponding
path and leaves the cryptographic root and asset identity unchanged.

\subsection{Proposed Approach}\label{sec:approach}

\CTDAP is parameterized by an abstract \emph{root-derivable framework}
(\autoref{def:rdf}) with three properties: deterministic derivation from
a single root entity, context-isolated key generation, and
authorization-bound revocation via destructible credential components.
Each control right is a dormant authorization path composed of:
\begin{itemize}[leftmargin=2em]
  \item One or more user-held authorization credentials, and
  \item One or more administrative authorization factors held by
        independent custodians.
\end{itemize}

The right takes effect only when all required factors are available at
once, and custodians release their factors after verifying an external
event (\autoref{sec:conditions}).

\subsection{Contributions}\label{sec:contributions}

The core contributions of this work are:

\begin{enumerate}[leftmargin=2em]
  \item \textbf{Formal Model.}
    We define the \CTDAP model with formal syntax
    (\autoref{sec:model}), specifying algorithms for path construction,
    activation, and revocation, parameterized by an abstract
    root-derivable framework.

  \item \textbf{Formal Security Analysis.}
    We define a threat model (\autoref{sec:threat-model}) and three
    security games (unauthorized control resistance, cross-path
    isolation, and stateless revocation) and prove security under
    standard cryptographic assumptions (\autoref{sec:security}).

  \item \textbf{Authorization-Bound Revocation.}
    We formalize an immediate, stateless revocation mechanism based
    on destructible authorization factors, which disables further
    activation without key rotation or asset migration while no path
    over the same root has been activated.

  \item \textbf{Concrete Instantiation and Evaluation.}
    We instantiate \CTDAP using the Atomic Cryptographic Entity
    Generative Framework (\ACEGF)~\cite{acegf-tr} and evaluate
    performance, showing sub-second activation latency with
    configurable security--performance trade-offs
    (\autoref{sec:performance}).

  \item \textbf{Application Scenarios.}
    We illustrate the method in four scenarios: regulated
    custody, inheritance planning, delegated control, and multi-party
    governance.
\end{enumerate}

\subsection{Organization}\label{sec:organization}

\autoref{sec:related} surveys related work.
\autoref{sec:model} defines the system model, trust assumptions, and formal
syntax.
\autoref{sec:construction} presents the authorization path construction
using \ACEGF.
\autoref{sec:activation} describes the condition-triggered activation and
revocation mechanisms.
\autoref{sec:security} provides the formal security analysis.
\autoref{sec:comparison} compares the proposed method with existing
approaches.
\autoref{sec:scenarios} presents illustrative application scenarios.
\autoref{sec:discussion} discusses practical considerations, limitations,
and future directions.
\autoref{sec:performance} presents performance benchmarks.
\autoref{sec:conclusion} concludes the paper.

\section{Related Work}\label{sec:related}

We review prior approaches to cryptographic asset control and position
our work relative to them.

\subsection{Single-Key and Hierarchical Deterministic Wallets}

Traditional cryptocurrency wallets equate control with possession of a
single private key or a hierarchically derived key tree
(BIP-32~\cite{bip32}, BIP-39~\cite{bip39}, BIP-44~\cite{bip44}).
These systems offer simplicity and full user autonomy but provide no
native support for conditional activation, delegation, or revocation.
Any control transition requires direct key transfer or asset migration,
both of which are irreversible and operationally costly.

\subsection{Multi-Signature and Threshold Signature Schemes}

Multi-signature wallets~\cite{goldfeder2015multisig} require $m$-of-$n$
signatures to authorize a transaction.  Threshold signature schemes
(TSS), including GG18~\cite{gg18}, GG20~\cite{gg20}, and
FROST~\cite{frost}, eliminate the need for on-chain multi-sig scripts
by distributing key generation and signing across $n$ parties with a
threshold $t$.

While TSS reduces on-chain footprint, both approaches require
persistent availability of all participating key shares, lack native
support for dormant (conditionally inactive) control rights, and
require resharing or migration when control conditions change.
Proactive secret sharing~\cite{proactivess} can
refresh shares periodically, but does not address conditional
activation tied to external events.

\subsection{Secret Sharing and Social Recovery}

Shamir's secret sharing~\cite{shamir79} enables $(t,n)$-threshold
reconstruction of a master secret.  Social recovery
wallets~\cite{wohrer2021social} apply this principle by
designating guardians who can collectively recover access.

These schemes conflate \emph{recovery} with \emph{ongoing control}:
once the secret is reconstructed, the reconstructing party gains full
and irrevocable access.  There is no mechanism for conditional
activation that preserves dormancy, nor for revoking a specific
guardian's capability without redistributing shares.

\subsection{Attribute-Based and Policy-Based Access Control}

Ciphertext-policy ABE (CP-ABE)~\cite{bethencourt2007cpabe} and
key-policy ABE (KP-ABE)~\cite{goyal2006kpabe} embed fine-grained
access policies into the encryption layer.  While conceptually related
to our notion of conditional activation, ABE schemes are designed for
\emph{data confidentiality} (who can decrypt a ciphertext) rather than
\emph{asset control} (who can derive a signing key or exercise
cryptographic authority over an asset).  Furthermore, revocation in ABE
typically requires re-encryption or key update
mechanisms~\cite{sahai2012abe-revocation}, which conflict with our
statelessness requirement.

\subsection{Smart-Contract-Based Enforcement}

Smart contracts~\cite{wood2014ethereum,szabo1997smart} enable
programmable, on-chain enforcement of control policies, including
multi-sig, time-locks, and conditional transfers.  However,
contract-based control depends on on-chain state mutation, is vulnerable
to oracle manipulation~\cite{adler2018astraea}, and ties control
semantics to a specific blockchain's execution model.  Control
transitions require explicit asset migration (e.g., transferring tokens
to a new contract address), which is neither stateless nor
chain-agnostic.

\subsection{Time-Lock Cryptography and Verifiable Delay Functions}

Time-lock puzzles~\cite{rivest1996timelock} and verifiable delay
functions (VDFs)~\cite{boneh2018vdf} provide time-conditioned secret
release.  While useful for temporal conditions, they are limited to
a single trigger type (elapsed time), cannot incorporate arbitrary
external events (e.g., legal recognition of inheritance), and are
irrevocable once the puzzle is generated.  Our method supports
arbitrary condition types and immediate revocation.

\subsection{MPC-Based Custody Solutions}

Institutional custody platforms (e.g.,~\cite{lindell2021mpc}) use
multi-party computation (MPC) to distribute signing authority without
reconstructing the private key.  While MPC custody preserves
non-disclosure of the full key, it requires persistent online
participation of MPC nodes, and policy changes require re-keying.
Our approach achieves conditional control through dormant paths that
can be activated or revoked without any online coordination among
authorization factor holders.

\subsection{Positioning of This Work}

\autoref{tab:related-comparison} summarizes the positioning of \CTDAP
relative to existing approaches.  \CTDAP separates asset identity from
condition-dependent control and supports dormant rights, conditional
activation, and stateless revocation without on-chain state or trusted
intermediaries.

\begin{table}[t]
\centering
\caption{Comparison of cryptographic asset control approaches.
  \checkmark = natively supported,
  $\circ$ = partially or with workarounds,
  $\times$ = not supported.
  For \CTDAP, the revocation columns hold only while no path over the
  same root has been activated; see \autoref{sec:sr}.}
\label{tab:related-comparison}
\small
\begin{tabular}{@{}lccccc@{}}
\toprule
\textbf{Approach} &
\rotatebox{60}{\textbf{Conditional}} &
\rotatebox{60}{\textbf{Dormant}} &
\rotatebox{60}{\textbf{Stateless}} &
\rotatebox{60}{\textbf{Immediate}} &
\rotatebox{60}{\textbf{No On-Chain}} \\
& \rotatebox{60}{\textbf{Activation}} &
\rotatebox{60}{\textbf{Rights}} &
\rotatebox{60}{\textbf{Revocation}} &
\rotatebox{60}{\textbf{Revocation}} &
\rotatebox{60}{\textbf{State}} \\
\midrule
Single-Key / HD Wallet    & $\times$ & $\times$ & $\times$ & $\times$ & \checkmark \\
Multi-Sig / TSS           & $\circ$  & $\times$ & $\times$ & $\circ$  & $\circ$ \\
Secret Sharing / Recovery & $\circ$  & $\times$ & $\times$ & $\times$ & \checkmark \\
CP-ABE / KP-ABE           & \checkmark & $\circ$ & $\times$ & $\circ$  & \checkmark \\
Smart Contracts           & \checkmark & $\circ$ & $\times$ & $\circ$  & $\times$ \\
Time-Lock / VDF           & $\circ$  & \checkmark & $\times$ & $\times$ & \checkmark \\
MPC Custody               & $\circ$  & $\times$ & $\times$ & $\circ$  & \checkmark \\
\textbf{\CTDAP (Ours)}    & \checkmark & \checkmark & \checkmark & \checkmark & \checkmark \\
\bottomrule
\end{tabular}
\end{table}

\section{System Model}\label{sec:model}

This section defines the abstract system model, trust assumptions, and
formal syntax for the \CTDAP method.

\subsection{Entities and Roles}\label{sec:entities}

The system involves the following entities:

\begin{itemize}[leftmargin=2em]
  \item \textbf{Asset Owner ($\mathcal{O}$):}
    The entity that holds the cryptographic root and establishes
    authorization paths.  $\mathcal{O}$ possesses the user-held
    credential $\UC$ and the sealed artifact $\SA$.

  \item \textbf{Custodian ($\mathcal{C}_i$):}
    An independent entity that holds and conditionally releases a
    single administrative authorization factor $\AF_i$.
    Custodians cannot unilaterally exercise asset control.

  \item \textbf{Beneficiary ($\mathcal{B}_j$):}
    An entity designated to receive control capability upon
    activation of authorization path $\AP_j$.  A beneficiary may
    also act as the asset owner in self-authorization scenarios.

  \item \textbf{Condition Verifier ($\mathcal{V}$):}
    An entity (or protocol) that verifies whether a predefined
    external condition is satisfied and provides a
    cryptographically signed attestation to the custodian.
    The verifier may be a trusted third party (e.g., notary,
    regulatory body) or a decentralized oracle
    network~\cite{adler2018astraea}.
\end{itemize}

\subsection{Trust Model and Assumptions}\label{sec:threat-model}

We formalize the trust assumptions that underpin the security of \CTDAP.

\begin{assumption}[Honest-but-Curious Custodians]\label{asm:custodian}
  Each custodian $\mathcal{C}_i$ is assumed to faithfully hold and
  conditionally release $\AF_i$ according to predefined release
  conditions.  However, custodians may attempt to learn information
  about other authorization factors or the root secret from the
  information available to them.  Custodians do not collude with
  each other or with unauthorized parties unless explicitly modeled.
\end{assumption}

\begin{assumption}[Bounded Adversary]\label{asm:adversary}
  The adversary $\Adv$ is a probabilistic polynomial-time (PPT)
  algorithm that may:
  \begin{enumerate}[label=(\roman*)]
    \item Compromise up to $t < n$ custodians for a given path with
      $n$ authorization factors (obtaining their $\AF_i$ values);
    \item Obtain the sealed artifact $\SA$ (e.g., by compromising
      storage);
    \item Observe activated control keys from \emph{other} paths
      $\AP_k$ ($k \neq j$);
    \item Interact with the condition verifier.
  \end{enumerate}
  The adversary \emph{cannot} simultaneously obtain \emph{all}
  components of $\Cred_{\mathrm{composite}}$ for a target path.
\end{assumption}

\begin{assumption}[Condition Verifier Integrity]\label{asm:verifier}
  The condition verifier $\mathcal{V}$ correctly evaluates the
  predefined condition and provides attestations only when the
  condition is genuinely satisfied.  For deployments requiring
  stronger guarantees, $\mathcal{V}$ may be instantiated as a
  decentralized oracle with multi-source aggregation to mitigate
  single-point-of-failure risks (see \autoref{sec:condition-verification}).
\end{assumption}

\begin{assumption}[Cryptographic Hardness]\label{asm:crypto}
  We assume:
  \begin{enumerate}[label=(\roman*)]
    \item The sealing mechanism (e.g., AES-256-GCM-SIV) provides
      both IND-CPA security (ciphertext indistinguishability) and
      INT-CTXT security (ciphertext integrity / authentication):
      decryption with an incorrect key returns $\bot$ with
      overwhelming probability~\cite{rfc8452};
      \label{asm:aead}
    \item HKDF-SHA256~\cite{rfc5869} is a pseudorandom function (PRF)
      in the random oracle model;
      \label{asm:prf}
    \item Argon2id~\cite{argon2} provides memory-hard key derivation
      with resistance to offline brute-force attacks;
      \label{asm:mh}
    \item The hash function $H$ used for path identifiers (e.g.,
      SHA-256) is collision-resistant: for any PPT adversary $\Adv$,
      $\Pr[\Adv \text{ finds } x \neq x' \text{ s.t.\ } H(x) =
      H(x')] \leq \negl(\secparam)$.
      \label{asm:cr}
  \end{enumerate}
\end{assumption}

\subsection{Root-Derivable Framework (Abstract Interface)}\label{sec:rdf}

\CTDAP is parameterized by a \emph{root-derivable framework}
satisfying the following interface:

\begin{definition}[Root-Derivable Framework]\label{def:rdf}
  A root-derivable framework $\Pi$ is a tuple of PPT algorithms
  $(\mathsf{Setup}, \mathsf{Seal}, \mathsf{Unseal}, \mathsf{Derive})$
  such that:

  \begin{itemize}[leftmargin=2em]
    \item $\mathsf{Setup}(1^\secparam) \to (\REV, \mathsf{params})$:
      Generates a 256-bit Root Entropy Value and public parameters.

    \item $\mathsf{Seal}(\mathsf{params}, \Cred_{\mathrm{composite}},
      \REV) \to \SA$:
      Encrypts $\REV$ under the composite credential, producing
      a sealed artifact.

    \item $\mathsf{Unseal}(\mathsf{params}, \SA,
      \Cred_{\mathrm{composite}}) \to \REV \cup \{\bot\}$:
      Reconstructs $\REV$ if and only if the correct composite
      credential is provided; otherwise returns $\bot$.

    \item $\mathsf{Derive}(\REV, \Ctx) \to \KeyAP$:
      Deterministically derives a context-isolated control key from
      $\REV$ and a context tuple $\Ctx$.
  \end{itemize}
\end{definition}

\begin{definition}[Required Security Properties]\label{def:rdf-security}
  The framework $\Pi$ must satisfy:
  \begin{enumerate}[label=\textbf{P\arabic*}]
    \item \textbf{Sealing Indistinguishability:}
      For any two roots $\REV_0, \REV_1$ of equal length, no PPT
      adversary without $\Cred_{\mathrm{composite}}$ can distinguish
      the two sealed artifacts
      \[
        \mathsf{Seal}(\mathsf{params},
        \Cred_{\mathrm{composite}}, \REV_b),
        \quad b \in \{0,1\}
      \]
      with non-negligible advantage.
      \label{prop:seal-ind}

    \item \textbf{Unsealing Correctness:}
      $\mathsf{Unseal}(\mathsf{params}, \SA, \Cred') = \bot$, except
      with negligible probability, for any
      $\Cred' \neq \Cred_{\mathrm{composite}}$ chosen by a PPT
      adversary.
      \label{prop:unseal-correct}

    \item \textbf{Context Isolation:}
      For distinct contexts $\Ctx_j \neq \Ctx_k$, the values
      $\mathsf{Derive}(\REV, \Ctx_j)$ and
      $\mathsf{Derive}(\REV, \Ctx_k)$ are computationally
      independent: no PPT adversary given one can
      distinguish the other from uniform with non-negligible
      advantage.
      \label{prop:ctx-isolation}
  \end{enumerate}
\end{definition}

\subsection{Formal Syntax}\label{sec:syntax}

\begin{definition}[\CTDAP Syntax]\label{def:ctdap-syntax}
  The \CTDAP method consists of the following algorithms:

  \begin{itemize}[leftmargin=2em]
    \item $\mathsf{PathSetup}(\Pi, 1^\secparam, \UC, \{\AF_i\}_{i=1}^{n})
      \to (\SA_{\mathrm{root}}, s, \mathsf{params})$:
      Establishes the root entity, generates a random salt
      $s \in \{0,1\}^{128}$, and seals the root under the composite
      credential.  Returns the sealed artifact, the public salt, and
      framework parameters.

    \item $\mathsf{PathDerive}(\Pi, \REV, \Ctx_j) \to \KeyAP[j]$:
      Derives the dormant control key for path $j$.  This is invoked
      only after successful unsealing.

    \item $\mathsf{Activate}(\Pi, \SA, \UC, \{\AF_i\}_{i=1}^{n},
      \Ctx_j) \to \KeyAP[j] \cup \{\bot\}$:
      Reconstructs $\REV$ via $\mathsf{Unseal}$, derives
      $\KeyAP[j]$, zeroizes $\REV$ from memory, and returns the
      control key.  Returns $\bot$ if any factor is missing or
      incorrect.

    \item $\mathsf{Revoke}(i) \to \{\mathsf{revoked}\}$:
      Destroys $\AF_i$ (e.g., secure erasure), making
      $\Cred_{\mathrm{composite}}$ irreconstructible for all paths
      depending on~$\AF_i$ unless a copy of $\AF_i$ was retained.
      This blocks further activation only
      (\autoref{sec:irreversibility}).
  \end{itemize}
\end{definition}

\begin{definition}[Authorization Path]\label{def:auth-path}
  The $j$-th authorization path $\AP_j$ is defined as a 4-tuple:
  \[
    \AP_j = (\SA_{\mathrm{root}},\; s_j,\; \UC_j,\; \mathcal{F}_j)
  \]
  where:
  \begin{itemize}[leftmargin=2em]
    \item $\SA_{\mathrm{root}}$: the sealed artifact (persistent
      ciphertext of $\REV_{\mathrm{root}}$; the root entropy value
      $\REV_{\mathrm{root}}$ itself is \emph{never} stored
      persistently; see \autoref{alg:path-setup}).
    \item $s_j \in \{0,1\}^{128}$: the per-path random salt for
      Argon2id.
    \item $\UC_j$: a path-specific user-held credential.
    \item $\mathcal{F}_j = \{\AF_{j,1}, \ldots, \AF_{j,n_j}\}$:
      the set of destructible administrative factors for path $j$.
  \end{itemize}
  Each path is sealed separately: in $\AP_j$, $\SA_{\mathrm{root}}$
  denotes the encryption of $\REV_{\mathrm{root}}$ under
  $\Cred_{\mathrm{composite},j}$ (\autoref{eq:cred-composite}), so $m$
  paths over one root store $m$ sealed artifacts of the same
  $\REV_{\mathrm{root}}$.  \autoref{alg:path-setup} establishes one
  path; a further path is added by unsealing $\REV_{\mathrm{root}}$ and
  sealing it again under the new path's composite credential and salt.
  Every path is $n_j$-of-$n_j$: activation needs all of its factors.
  A path is \emph{dormant} when at least one $\AF_{j,i}$ is withheld.
  A path is \emph{active} when all components are simultaneously
  available, enabling ephemeral reconstruction of
  $\REV_{\mathrm{root}}$ via $\mathsf{Unseal}$ and subsequent
  derivation of $\KeyAP[j]$.
  A path is \emph{revoked} when at least one $\AF_{j,i}$ is
  destroyed with no copy retained.  A revoked path cannot be activated
  again; revocation does not remove control from a party that has
  already activated a path (\autoref{sec:irreversibility}).
\end{definition}

\section{Concrete Instantiation: ACE-GF}\label{sec:primitive}

We instantiate the root-derivable framework $\Pi$ using the
\emph{Atomic Cryptographic Entity Generative Framework}
(\ACEGF)~\cite{acegf-tr}, a seed-storage-free identity
construction.  \ACEGF satisfies all three required properties
(\ref{prop:seal-ind}--\ref{prop:ctx-isolation}) of
\autoref{def:rdf-security}.

\subsection{Key Properties of ACE-GF}\label{sec:acegf-properties}

\begin{enumerate}[leftmargin=2em]
  \item \textbf{Deterministic Root-Derived Control Capabilities:}
    \ACEGF establishes $\REV$ as the cryptographic root entity,
    which exists only ephemerally in memory and is never persistently
    stored.  Control keys are deterministically derived from $\REV$
    via HKDF-SHA256 with explicit context
    encoding~\cite{rfc5869,acegf-tr}.

  \item \textbf{Authorization-Bound Revocation with Destructible
    Components:}
    \ACEGF binds access to $\REV$ through a sealed artifact ($\SA$)
    encrypted under the composite credential.  Revocation is achieved
    by destroying any designated credential component, making further
    $\REV$ reconstruction from $\SA$ impossible~\cite{acegf-tr}.

  \item \textbf{Seed-Storage-Free Operation:}
    \ACEGF eliminates persistent storage of master secrets (e.g.,
    BIP-39 mnemonics).  Only $\SA$, an encrypted representation of
    $\REV$, is stored at rest.

  \item \textbf{Nonce-Misuse Resistance:}
    The sealing mechanism uses AES-256-GCM-SIV~\cite{rfc8452},
    providing nonce-misuse-resistant authenticated encryption under a
    fixed public nonce.  In the general \ACEGF construction each sealed
    artifact carries its own random Argon2id salt, so sealing is
    randomized; only the 24-word wallet encoding profile of
    \ACEGF~\cite{acegf-tr}, which fixes the salt, seals
    deterministically for a given credential.

  \item \textbf{Memory-Hard Credential Derivation:}
    Credential composition is processed via Argon2id~\cite{argon2},
    ensuring resistance to offline brute-force attacks.
\end{enumerate}

\subsection{Relevance to CT-DAP}\label{sec:acegf-relevance}

\ACEGF's separation of \emph{identity} (rooted in $\REV$) and
\emph{authorization} (governed by destructible credentials) is the
basis of dormant authorization paths:
\begin{itemize}[leftmargin=2em]
  \item $\REV$ serves as the single cryptographic control origin for
    all authorization paths.
  \item \ACEGF's context-isolated key derivation, using
    structured context tuples $\Ctx =
    (\mathsf{AlgID},\allowbreak \mathsf{Domain},\allowbreak
    \mathsf{Index})$, enables parallel, isolated
    authorization paths.
  \item \ACEGF's authorization-bound revocation mechanism supplies
    the destructible authorization factors on which \CTDAP relies.
\end{itemize}

\section{Authorization Path Construction}\label{sec:construction}

We now construct dormant authorization paths, the central abstraction
of \CTDAP, using \ACEGF.

\subsection{Construction Overview}

Each authorization path is a composite structure of user-held
credentials and custodian-managed destructible authorization factors,
anchored to a single $\REV$ in \ACEGF.  The construction proceeds in
three steps, formalized in \autoref{alg:path-setup}.

\begin{algorithm}[t]
\caption{$\mathsf{PathSetup}$: Establish Root Entity and Dormant Paths}
\label{alg:path-setup}
\begin{algorithmic}[1]
\Require Security parameter $1^\secparam$; user credential $\UC$; factors $\{\AF_i\}_{i=1}^{n}$
\Ensure Sealed artifact $\SA_{\mathrm{root}}$; per-path salt
  $s \in \{0,1\}^{128}$; public parameters $\mathsf{params}$

\Statex \textbf{Step 1: Establish Root Entity}
\State $(\REV_{\mathrm{root}}, \mathsf{params}) \gets
  \ACEGF.\mathsf{Setup}(1^\secparam)$
\State $s \getsr \{0,1\}^{128}$
  \Comment{Per-path random salt}
\State $\Cred_{\mathrm{composite}} \gets
  \mathsf{Argon2id}(\UC \concat \AF_1 \concat \cdots \concat \AF_n
  ;\; s;\; m, t, p)$
  \Comment{Memory-hard derivation}
\State $\SA_{\mathrm{root}} \gets
  \ACEGF.\mathsf{Seal}(\mathsf{params},
  \Cred_{\mathrm{composite}}, \REV_{\mathrm{root}})$

\Statex
\Statex \textbf{Step 2: (Optional) Verify Sealed Artifact}
  \Comment{Self-test before zeroizing root}
\State $\REV' \gets \ACEGF.\mathsf{Unseal}(\mathsf{params},
  \SA_{\mathrm{root}}, \Cred_{\mathrm{composite}})$
\State \textbf{assert} $\REV' = \REV_{\mathrm{root}}$
  \Comment{Abort if seal/unseal round-trip fails}
\State $\mathsf{Zeroize}(\REV', \REV_{\mathrm{root}})$
  \Comment{Root exists only ephemerally}

\Statex
\Statex \textbf{Step 3: Seal Path for Dormancy}
\State Encode $\SA_{\mathrm{root}}$ as a phrase over the BIP-39
  wordlist (ciphertext encoding; not a BIP-39 seed)
\State Store salt $s$ alongside $\SA_{\mathrm{root}}$ in cleartext
\State Distribute $\AF_i$ to respective custodians $\mathcal{C}_i$
  via secure channels
\State \Return $(\SA_{\mathrm{root}}, s, \mathsf{params})$

\Statex
\Statex \textit{Note: Context-isolated control keys $\KeyAP[j]$ are
  \textbf{not} derived at setup time.  They are derived on-demand
  during $\mathsf{Activate}$ (\autoref{alg:activate}), which
  ephemerally reconstructs $\REV_{\mathrm{root}}$ from
  $\SA_{\mathrm{root}}$.}
\end{algorithmic}
\end{algorithm}

\subsection{Credential Composition}\label{sec:cred-composition}

The composite credential for path $j$ is constructed by processing
the user credential and all administrative factors through a
memory-hard key derivation function with a \emph{path-specific random
salt} $s_j$:
\begin{equation}\label{eq:cred-composite}
  \Cred_{\mathrm{composite},j} =
    \mathsf{Argon2id}\bigl(
      \underbrace{\UC \concat \AF_1 \concat \cdots \concat \AF_n}_{\text{password}},\;
      s_j,\;
      m,\, t,\, p
    \bigr)
\end{equation}
where $\concat$ denotes concatenation in which each credential
component is prefixed by its length, so that distinct component tuples
give distinct passwords (a delimiter byte such as \texttt{0x00} does
not, because random factors may contain it), and:
\begin{itemize}[leftmargin=2em]
  \item $s_j \getsr \{0,1\}^{128}$ is a per-path random salt, stored
    alongside $\SA_{\mathrm{root}}$ in cleartext.  The salt prevents
    cross-target amortization: an adversary attacking path $j$ cannot
    reuse work (e.g., precomputed tables or parallel GPU batches) to
    simultaneously attack path $k \neq j$, even if the underlying
    credential components are identical.
  \item $(m, t, p)$ are the Argon2id memory, time, and parallelism
    parameters, chosen per deployment policy (e.g., $m = 2^{19}$~KiB,
    $t = 3$, $p = 4$ per RFC~9106~\cite{rfc9106}).
\end{itemize}

The use of Argon2id~\cite{argon2} as a memory-hard key derivation
function ensures resistance to offline brute-force attacks, even
if the adversary obtains $\SA_{\mathrm{root}}$ and $s_j$, provided
the missing credential component has sufficient min-entropy
($\geq 128$~bits).

\subsection{Sealed Artifact as Dormant Mnemonic}\label{sec:sealed-artifact}

Unlike traditional systems where a mnemonic directly represents the
plaintext master seed (e.g., BIP-39), the sealed artifact
($\SA$) is encoded as a phrase over the BIP-39 wordlist that carries the
\emph{ciphertext} of $\REV_{\mathrm{root}}$ rather than the plaintext
entropy.  Only the BIP-39 wordlist is reused: the phrase is not a
BIP-39 seed, and its length follows from the ciphertext length.

Without the requisite composite credential, $\SA$ is functionally
indistinguishable from random noise and cannot be used to derive any
control capabilities.  AES-GCM-SIV makes the fixed public nonce safe.
Sealing is not deterministic in the general construction: each sealed
artifact carries its own random salt, so sealing the same root under
the same credential twice gives two different artifacts.  Only the
24-word wallet encoding profile of \ACEGF~\cite{acegf-tr}, which fixes
the salt, maps a credential and a root to a single artifact.

\subsection{Isolation Guarantee Across Paths}\label{sec:isolation}

By \ACEGF's context-isolation property (\ref{prop:ctx-isolation}),
control keys derived from the same $\REV_{\mathrm{root}}$ but with
distinct context values are computationally independent:
\begin{equation}\label{eq:isolation}
  \forall j \neq k, \quad
  \mathrm{Adv}_\Adv\bigl(
    \KeyAP[j] \mid \KeyAP[k]
  \bigr) \leq \negl(\secparam)
\end{equation}

Exposure of a control key $\KeyAP[k]$ therefore reveals nothing about
the control keys of other paths.  Activation is different: the
activating party reconstructs $\REV_{\mathrm{root}}$ itself
(\autoref{alg:activate}) and can derive the control key of every path
over the same root.  Isolation between paths thus holds against key
exposure, not against a party that has activated any path; separating
the beneficiaries of one root requires an independent $\REV$ per path,
or unsealing and signing inside a TEE or MPC boundary that never
exports $\REV_{\mathrm{root}}$.

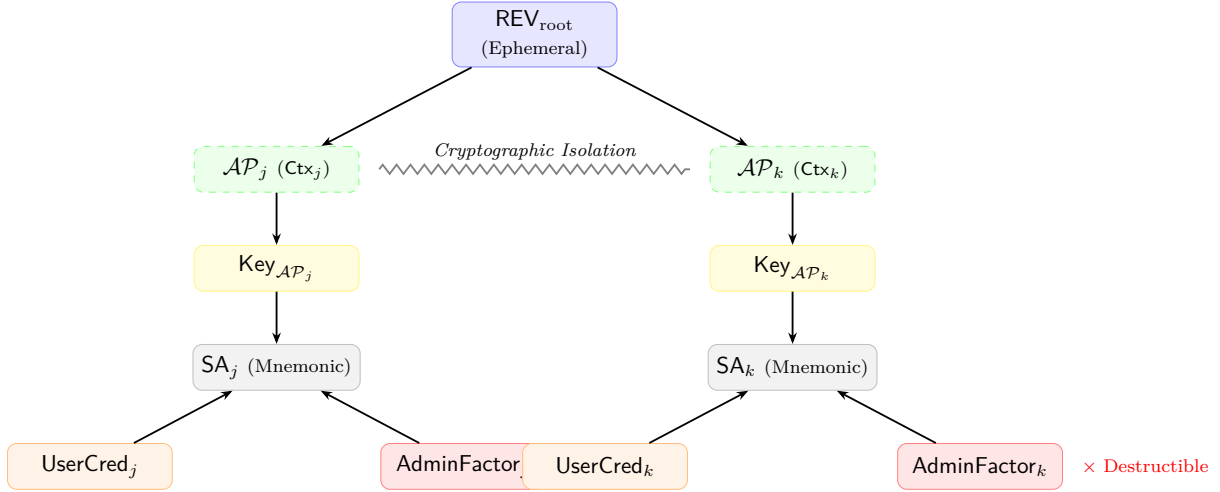
\begin{figure}[t]
\centering
\resizebox{\columnwidth}{!}{%
\begin{tikzpicture}[
  node distance=1.2cm and 2.5cm,
  block/.style={rectangle, draw, rounded corners, minimum width=2.5cm,
    minimum height=0.7cm, align=center, font=\small},
  root/.style={block, fill=blue!10, draw=blue!50},
  path/.style={block, fill=green!8, draw=green!50, dashed},
  key/.style={block, fill=yellow!15, draw=yellow!60},
  sa/.style={block, fill=gray!10, draw=gray!50},
  factor/.style={block, fill=red!10, draw=red!50},
  cred/.style={block, fill=orange!10, draw=orange!50},
  arr/.style={-{Stealth[length=5pt]}, thick},
  iso/.style={decorate, decoration={zigzag, amplitude=2pt, segment length=6pt},
    draw=gray, thick},
]
  \node[root] (rev) {$\REV_{\mathrm{root}}$\\{\scriptsize(Ephemeral)}};

  \node[path, below left=1.2cm and 1.4cm of rev] (apj)
    {$\AP_j$ \scriptsize($\Ctx_j$)};
  \node[path, below right=1.2cm and 1.4cm of rev] (apk)
    {$\AP_k$ \scriptsize($\Ctx_k$)};

  \node[key, below=0.8cm of apj] (keyj) {$\KeyAP[j]$};
  \node[key, below=0.8cm of apk] (keyk) {$\KeyAP[k]$};

  \node[sa, below=0.8cm of keyj] (saj) {$\SA_j$ {\scriptsize(Mnemonic)}};
  \node[sa, below=0.8cm of keyk] (sak) {$\SA_k$ {\scriptsize(Mnemonic)}};

  \node[cred, below left=0.8cm and 0.3cm of saj] (ucj) {\small$\UC_j$};
  \node[factor, below right=0.8cm and 0.3cm of saj] (afj)
    {\small$\AF_{j}$};
  \node[cred, below left=0.8cm and 0.3cm of sak] (uck) {\small$\UC_k$};
  \node[factor, below right=0.8cm and 0.3cm of sak] (afk)
    {\small$\AF_{k}$};

  \draw[arr] (rev) -- (apj);
  \draw[arr] (rev) -- (apk);
  \draw[arr] (apj) -- (keyj);
  \draw[arr] (apk) -- (keyk);
  \draw[arr] (keyj) -- (saj);
  \draw[arr] (keyk) -- (sak);
  \draw[arr] (ucj) -- (saj);
  \draw[arr] (afj) -- (saj);
  \draw[arr] (uck) -- (sak);
  \draw[arr] (afk) -- (sak);

  \draw[iso] ($(apj.east)+(0.3,0)$) -- ($(apk.west)+(-0.3,0)$)
    node[midway, above, font=\scriptsize\itshape] {Cryptographic Isolation};

  \node[font=\scriptsize\color{red}, right=0.15cm of afk]
    {$\times$ Destructible};
\end{tikzpicture}%
}
\caption{Isolation and redundancy of authorization paths.
Each path uses a unique context tuple and independent
authorization factors.  Destroying a factor revokes the
corresponding path without affecting the root or other paths.}
\label{fig:isolation}
\end{figure}

\section{Conditional Triggering and Activation}\label{sec:activation}

Administrative authorization factors are not persistently available.
Their release is conditioned on verification of predefined external
events.

\subsection{Trigger Conditions}\label{sec:conditions}

Supported condition types include, but are not limited to:
\begin{itemize}[leftmargin=2em]
  \item \textbf{Explicit User Consent:} Direct authorization by the
    asset owner (e.g., signed request).
  \item \textbf{Legal Events:} Legally recognized inheritance or
    succession events, verified by a certified entity (e.g., notary,
    probate court).
  \item \textbf{Temporal Conditions:} Time-based or jurisdiction
    deadlines, verifiable via trusted timestamps or on-chain
    time-locks.
  \item \textbf{Regulatory Actions:} Suspension or compliance-driven
    release by a supervisory authority.
\end{itemize}

\subsection{Condition Verification}\label{sec:condition-verification}

Condition verification may be instantiated via two complementary
mechanisms:

\paragraph{Trusted Third-Party (TTP) Verifiers.}
For legally binding conditions, certified entities (e.g., notaries,
regulatory bodies) act as verifiers, providing cryptographically
signed condition-fulfillment proofs to custodians.

\paragraph{Decentralized Oracles.}
For transparent, tamper-resistant conditions (e.g., time locks,
cross-chain events), decentralized oracle
networks~\cite{adler2018astraea} can replace TTPs, providing
condition verification without single points of failure.  Oracle
manipulation risks are mitigated via multi-source data aggregation.

\subsection{Activation Protocol}\label{sec:activation-protocol}

The activation protocol is formalized in \autoref{alg:activate}.

\begin{algorithm}[t]
\caption{$\mathsf{Activate}$: Condition-Triggered Path Activation}
\label{alg:activate}
\begin{algorithmic}[1]
\Require Sealed artifact $\SA_{\mathrm{root}}$; user credential
  $\UC$; released factors $\{\AF_i\}_{i=1}^{n}$; target context
  $\Ctx_j$; public parameters $\mathsf{params}$
\Ensure Control key $\KeyAP[j]$ or $\bot$

\Statex \textbf{Step 1: Factor Release}
\For{each custodian $\mathcal{C}_i$}
  \State Verify condition via $\mathcal{V}$
  \State $\AF_i \gets \ACEGF.\mathsf{Unseal}(\mathsf{params},
    \SA_{\mathcal{C}_i}, \Cred_{\mathcal{C}_i})$
    \Comment{Custodian releases factor}
  \State Transmit $\AF_i$ via secure channel (e.g., TLS~1.3 with
    certificate pinning)
\EndFor

\Statex
\Statex \textbf{Step 2: Composite Credential Reconstruction}
\State $\Cred_{\mathrm{composite}} \gets
  \mathsf{Argon2id}(\UC \concat \AF_1 \concat \cdots \concat \AF_n
  ;\; s;\; m, t, p)$
  \Comment{Salt $s$ read from $\SA$ metadata}

\Statex
\Statex \textbf{Step 3: REV Reconstruction and Key Derivation}
\State $\REV_{\mathrm{root}} \gets
  \ACEGF.\mathsf{Unseal}(\mathsf{params}, \SA_{\mathrm{root}},
  \Cred_{\mathrm{composite}})$
\If{$\REV_{\mathrm{root}} = \bot$}
  \State \Return $\bot$ \Comment{Credential mismatch}
\EndIf
\State $\KeyAP[j] \gets
  \ACEGF.\mathsf{Derive}(\REV_{\mathrm{root}}, \Ctx_j)$
\State $\mathsf{Zeroize}(\REV_{\mathrm{root}})$
  \Comment{Maintain ephemeral root property}
\State \Return $\KeyAP[j]$
\end{algorithmic}
\end{algorithm}

\subsection{Revocation Protocol}\label{sec:revocation-protocol}

Revocation is formalized in \autoref{alg:revoke}.

\begin{algorithm}[t]
\caption{$\mathsf{Revoke}$: Authorization-Bound Path Revocation}
\label{alg:revoke}
\begin{algorithmic}[1]
\Require Target authorization factor index $i$; custodian
  $\mathcal{C}_i$
\Ensure Path(s) depending on $\AF_i$ cannot be activated again

\State $\mathcal{C}_i$ performs secure erasure of $\AF_i$
  (e.g., zeroization of storage media)
\State \Comment{$\Cred_{\mathrm{composite}}$ is now
  irreconstructible unless a copy of $\AF_i$ was retained}
\State \Comment{$\REV_{\mathrm{root}}$ cannot be unsealed for
  path(s) requiring $\AF_i$}
\State Log destruction event for auditability
  (e.g., timestamped attestation)
\end{algorithmic}
\end{algorithm}

\subsection{Irreversibility of Activation}\label{sec:irreversibility}

Once activated, the authorization path remains usable as long as
$\Cred_{\mathrm{composite}}$ is complete.  To disable the path, a
custodian destroys their $\AF_i$ (e.g., secure erasure via
zeroization), rendering $\Cred_{\mathrm{composite}}$ permanently
incomplete.  This prevents any \emph{further} activation of the path.
It does not affect a party that has already activated any path over
$\REV_{\mathrm{root}}$: that party has held $\REV_{\mathrm{root}}$ and
can derive the control key of every path, so removing its control
requires moving the assets to keys under a fresh $\REV$.  Moving to a
new context suffices only when a single control key, not
$\REV_{\mathrm{root}}$, was exposed.

\section{Security Analysis}\label{sec:security}

This section provides formal security analysis of \CTDAP.  We first
specify the canonical encoding used by $\mathsf{Derive}$, then state
three auxiliary lemmas that capture the individual cryptographic
advantages, and finally prove the main security theorems by composing
these lemmas with explicit advantage bounds.

\subsection{Canonical Context Encoding for \texorpdfstring{$\mathsf{Derive}$}{Derive}}
\label{sec:ctx-encoding}

To ensure unambiguous domain separation across implementations, we
require a \emph{canonical} byte-level encoding of the context tuple
passed to HKDF\@.  Specifically, for a context
$\Ctx_j = (\mathsf{AlgID},\, \mathsf{Domain},\, \mathsf{Index})$,
the $\mathsf{Derive}$ function computes:
\begin{align}\label{eq:derive-canonical}
  \mathsf{PRK} &= \mathrm{HKDF\text{-}Extract}(
    \mathsf{salt} = \texttt{"CT-DAP-v1"},\;
    \mathsf{IKM} = \REV_{\mathrm{root}}) \\
  \KeyAP[j] &= \mathrm{HKDF\text{-}Expand}(
    \mathsf{PRK},\;
    \mathsf{info} = \mathsf{Encode}(\Ctx_j),\;
    L)
\end{align}
where the info string is the length-prefixed concatenation:
\[
  \mathsf{Encode}(\Ctx_j) =
    \langle|\mathsf{AlgID}|\rangle_{2}
    \concat \mathsf{AlgID}
    \concat \langle|\mathsf{Domain}|\rangle_{2}
    \concat \mathsf{Domain}
    \concat \langle\mathsf{Index}\rangle_{4}
\]
with $\langle \cdot \rangle_k$ denoting big-endian encoding into $k$
bytes.  The fixed-width length prefixes ensure that no two distinct
tuples $(\mathsf{AlgID}, \mathsf{Domain}, \mathsf{Index}) \neq
(\mathsf{AlgID}', \mathsf{Domain}', \mathsf{Index}')$ produce the
same info string, preventing domain-separation collisions.  The static
salt \texttt{"CT-DAP-v1"} binds the derivation to this protocol
version; future versions increment the version tag.

\subsection{Auxiliary Lemmas}\label{sec:aux-lemmas}

We isolate the three cryptographic advantages that appear in the main
proofs.

\begin{lemma}[Sealing Indistinguishability]\label{lem:seal}
  For any two roots $\REV_0, \REV_1$ of equal length, let
  $\SA_b = \ACEGF.\mathsf{Seal}(\mathsf{params},
  \Cred_{\mathrm{composite}}, \REV_b)$ for $b \getsr \{0,1\}$.
  For any PPT $\Adv$ that does not know
  $\Cred_{\mathrm{composite}}$:
  \[
    \mathrm{Adv}_{\Adv}^{\mathrm{SEAL}}(\secparam)
    \;=\;
    \bigl|\Pr[\Adv(\SA_b) = b] - \tfrac{1}{2}\bigr|
    \;\leq\;
    \mathrm{Adv}_{\Adv,\,\mathrm{AES\text{-}GCM\text{-}SIV}}^{\mathrm{IND\text{-}CPA}}(\secparam)
  \]
  where the right-hand side is the IND-CPA advantage against
  AES-256-GCM-SIV~\cite{rfc8452}.
\end{lemma}

\begin{proof}
  We reduce to the IND-CPA security of AES-256-GCM-SIV.

  \emph{Reduction.}
  Given a PPT distinguisher $\Adv$ for the sealing experiment,
  construct an IND-CPA adversary $\mathcal{B}$ as follows:
  \begin{enumerate}[label=(\roman*), leftmargin=2em]
    \item $\mathcal{B}$ receives the AEAD public parameters and
      an encryption oracle $\mathsf{Enc}_K(\cdot)$ under an unknown
      key $K$ chosen uniformly by the IND-CPA challenger.
    \item $\mathcal{B}$ sets $K = \Cred_{\mathrm{composite}}$
      (the key is unknown to $\mathcal{B}$).  It chooses two
      equal-length roots $\REV_0, \REV_1$ and submits them as
      challenge plaintexts $(m_0, m_1) = (\REV_0, \REV_1)$.
    \item The IND-CPA challenger returns $c^* = \mathsf{Enc}_K(
      \REV_b)$ for a uniformly random $b \getsr \{0,1\}$.
    \item $\mathcal{B}$ sets $\SA^* \gets c^*$ and forwards it
      to $\Adv$.
    \item $\Adv$ outputs a guess $b'$.  $\mathcal{B}$ outputs $b'$.
  \end{enumerate}

  \emph{Analysis.}
  By construction, the view of $\Adv$ when interacting with
  $\mathcal{B}$ is identical to the sealing experiment: $\SA^*$
  is the AEAD encryption of $\REV_b$ under
  $\Cred_{\mathrm{composite}}$.  Therefore
  $\Pr[\mathcal{B} \text{ wins}] = \Pr[\Adv(\SA_b) = b]$, and
  \[
    \mathrm{Adv}_{\Adv}^{\mathrm{SEAL}}(\secparam)
    = \bigl|\Pr[\Adv(\SA_b) = b] - \tfrac{1}{2}\bigr|
    = \mathrm{Adv}_{\mathcal{B},\,\mathrm{AES\text{-}GCM\text{-}SIV}}^{\mathrm{IND\text{-}CPA}}(\secparam)
  \]
  which is $\negl(\secparam)$ by \autoref{asm:crypto}\ref{asm:aead}.
  \qedhere
\end{proof}

\begin{lemma}[Credential Brute-Force Resistance]\label{lem:mh}
  Let the composite credential be
  \[
    \Cred_{\mathrm{composite}} = \mathsf{Argon2id}\!\bigl(
    \UC \!\concat\! \AF_1 \!\concat\! \cdots
    \!\concat\! \AF_n;\; s;\; m,t,p\bigr)
  \]
  where the components unknown to $\Adv$ have min-entropy
  $\geq \kappa$ bits conditioned on the components $\Adv$ knows,
  and $s$ is a 128-bit per-path random salt.  For any adversary
  $\Adv$ with at-most-$Q$ queries to an Argon2id evaluation oracle
  (each costing $\geq m$ bytes of memory):
  \[
    \mathrm{Adv}_{\Adv}^{\mathrm{MH}}(\secparam)
    \;\leq\; \frac{Q}{2^{\kappa}}
  \]
  The per-path salt ensures that work performed against target
  $(s_j, \SA_j)$ provides no non-trivial advantage against a
  distinct target $(s_k, \SA_k)$ with $s_k \neq s_j$: the
  adversary cannot amortize queries across targets and must
  restart with an independent budget of $Q$ evaluations per
  target.
\end{lemma}

\begin{proof}
  We model Argon2id as a random oracle $\mathcal{O}$ that, on
  each fresh query $(x, s, m, t, p)$, returns a uniformly random
  output in $\{0,1\}^{256}$.

  \emph{Single-target bound.}
  The adversary $\Adv$ knows at most $n$ of the $n+1$
  credential components $\{\UC, \AF_1, \ldots, \AF_n\}$.  The
  missing components have min-entropy $\geq \kappa$ given the
  known ones, so each
  query to $\mathcal{O}$ with a guessed value matches the true
  $\Cred_{\mathrm{composite}}$ with probability at most
  $2^{-\kappa}$.  Over $Q$ queries, a union bound gives
  \[
    \Pr[\exists\, i \leq Q :
      \mathcal{O}(x_i, s, m, t, p)
      = \Cred_{\mathrm{composite}}]
    \;\leq\; \frac{Q}{2^{\kappa}}
  \]

  \emph{Memory-hardness cost.}
  Each Argon2id evaluation requires $\geq m$ bytes of working
  memory~\cite{argon2,rfc9106}.  An adversary with total memory
  $M$ can execute at most $\lfloor M/m \rfloor$ parallel
  evaluations per time step, bounding the achievable query rate
  and making large $Q$ resource-prohibitive.

  \emph{Cross-target independence.}
  The random salt $s_j$ is drawn independently for each path.
  Since $\mathcal{O}$ is modelled as a random oracle, queries
  with salt $s_j$ and salt $s_k \neq s_j$ produce independent
  outputs.  Hence work performed against target $(s_j, \SA_j)$
  yields no non-trivial advantage against $(s_k, \SA_k)$: the
  adversary must spend an independent budget of $Q$ evaluations
  per target.
  \qedhere
\end{proof}

\begin{remark}[On the random oracle model for Argon2id]
\label{rem:argon2-rom}
  The proof above models Argon2id as a random oracle, which is a
  standard idealization shared by the Argon2 specification
  itself~\cite{argon2} and RFC~9106~\cite{rfc9106}.  In practice,
  Argon2id is \emph{not} a random oracle; its security rests on
  \emph{memory-hardness}, the property that any algorithm
  computing the function faster than the honest evaluator must use
  comparable memory.  Formalizing memory-hardness in the standard
  model requires graph-pebbling-based
  frameworks~\cite{alwen2017scrypt}, which yield asymptotically
  weaker but more realistic bounds of the form
  $\mathrm{cost}(\Adv) \geq \Omega(m \cdot t)$ per evaluation.

  Our ROM-based bound $Q / 2^{\kappa}$ should therefore be
  interpreted as an \emph{upper bound on the online query
  advantage}, complemented by the economic argument that each
  query costs $\geq m$ bytes of memory.  A full standard-model
  analysis of Argon2id's memory-hardness is an independent
  research problem beyond the scope of this work; we refer the
  reader to Alwen and Blocki~\cite{alwen2017scrypt} for the
  state-of-the-art treatment.
\end{remark}

\begin{lemma}[PRF Security of $\mathsf{Derive}$]\label{lem:prf}
  Let $\mathsf{Derive}$ be instantiated as HKDF-SHA256 per
  \autoref{sec:ctx-encoding}.  For any PPT $\Adv$ making at most
  $q$ queries with distinct canonical info strings:
  \[
    \mathrm{Adv}_{\Adv,\,\mathrm{HKDF}}^{\mathrm{PRF}}(\secparam)
    \;\leq\; q \cdot \mathrm{Adv}_{\mathrm{HMAC\text{-}SHA256}}^{\mathrm{PRF}}(\secparam)
  \]
  In the random oracle model for SHA-256, the right-hand side is
  $\negl(\secparam)$~\cite{rfc5869}.
\end{lemma}

\begin{proof}
  We reduce to the PRF security of HMAC-SHA256.

  \emph{Structure of HKDF.}
  HKDF consists of two phases~\cite{rfc5869}:
  $\mathsf{PRK} = \mathrm{HKDF\text{-}Extract}(\mathsf{salt},
  \mathsf{IKM})$ and
  $\mathsf{OKM} = \mathrm{HKDF\text{-}Expand}(\mathsf{PRK},
  \mathsf{info}, L)$.
  For a fixed $\mathsf{PRK}$, $\mathrm{HKDF\text{-}Expand}$ is
  an iterated application of $\mathrm{HMAC\text{-}SHA256}$ keyed
  by $\mathsf{PRK}$.  It is therefore a PRF from the info space
  to $\{0,1\}^L$, with advantage bounded by the PRF advantage of
  HMAC-SHA256.

  \emph{Hybrid argument.}
  Consider a sequence of $q+1$ hybrid games
  $H_0, H_1, \ldots, H_q$.  In $H_0$, all $q$ outputs
  $\KeyAP[1], \ldots, \KeyAP[q]$ are computed honestly via
  HKDF-Expand.  In $H_i$, the first $i$ outputs are replaced by
  independent uniform random strings, while the remaining
  $q - i$ are computed honestly.  In $H_q$, all outputs are
  uniformly random.

  The canonical encoding (\autoref{sec:ctx-encoding}) ensures
  that distinct contexts $\Ctx_j \neq \Ctx_k$ produce distinct
  info strings $\mathsf{Encode}(\Ctx_j) \neq
  \mathsf{Encode}(\Ctx_k)$.  Therefore each transition
  $H_{i-1} \to H_i$ replaces one PRF output on a fresh input
  with a random string, which is indistinguishable by
  $\mathrm{Adv}_{\mathrm{HMAC\text{-}SHA256}}^{\mathrm{PRF}}$.
  Summing over $q$ transitions:
  \[
    \mathrm{Adv}_{\Adv,\,\mathrm{HKDF}}^{\mathrm{PRF}}(\secparam)
    \;\leq\;
    q \cdot
    \mathrm{Adv}_{\mathrm{HMAC\text{-}SHA256}}^{\mathrm{PRF}}(\secparam)
  \]
  In the random oracle model for SHA-256, the right-hand side is
  $\negl(\secparam)$ for polynomial $q$.
  \qedhere
\end{proof}

\subsection{Unauthorized Control Resistance}\label{sec:ucr}

\begin{definition}[UCR Game]\label{def:ucr-game}
  The \emph{Unauthorized Control Resistance} game
  $\mathsf{Game}_{\mathrm{UCR}}^{\Adv}(\secparam)$ is parameterized
  by a \emph{challenge context} $\Ctx^*$ chosen uniformly from the
  set of valid path contexts $\{\Ctx_1, \ldots, \Ctx_m\}$, and
  proceeds as follows:
  \begin{enumerate}[label=\arabic*., leftmargin=2em]
    \item \textbf{Setup.}
      The challenger runs $\mathsf{PathSetup}$ to generate
      $(\SA_{\mathrm{root}}, s, \mathsf{params})$ with credential
      components $(\UC, \AF_1, \ldots, \AF_n)$ and path contexts
      $\{\Ctx_1, \ldots, \Ctx_m\}$.  The challenger computes
      $\KeyAP[^*] \!=\!
       \mathsf{Derive}(\REV_{\mathrm{root}}, \Ctx^*)$.
    \item \textbf{Corruption.}
      $\Adv$ chooses a proper subset
      $S \subsetneq \{\UC, \AF_1, \ldots, \AF_n\}$
      (i.e., at least one component is withheld).
      The challenger reveals all components in~$S$.
    \item \textbf{Public information.}
      $\Adv$ also receives $\SA_{\mathrm{root}}$, $\mathsf{params}$,
      $s$, the set of context labels $\{\Ctx_j\}_{j=1}^{m}$, and
      oracle access to $\mathsf{Derive}(\REV_{\mathrm{root}},
      \cdot)$ for \emph{non-challenge} contexts
      $\Ctx_j \neq \Ctx^*$ (modelling legitimate key use on other
      paths).
    \item \textbf{Output.}
      $\Adv$ outputs a candidate key $K^*$.
    \item \textbf{Win condition.}
      $\Adv$ wins if $K^* = \KeyAP[^*]$.
  \end{enumerate}
\end{definition}

\begin{theorem}[Unauthorized Control Resistance]\label{thm:ucr}
  Under \autoref{asm:crypto}, for any PPT adversary $\Adv$:
  \[
    \Pr\bigl[\mathsf{Game}_{\mathrm{UCR}}^{\Adv}(\secparam) = 1\bigr]
    \;\leq\;
    \mathrm{Adv}^{\mathrm{SEAL}}(\secparam)
    \;+\; \mathrm{Adv}^{\mathrm{MH}}(\secparam)
    \;+\; \mathrm{Adv}^{\mathrm{PRF}}(\secparam)
    \;+\; 2^{-L}
  \]
  where the first three terms correspond to \autoref{lem:seal},
  \autoref{lem:mh}, and \autoref{lem:prf} respectively, $L$ is the
  control-key length, and $\kappa$ in \autoref{lem:mh} is the
  min-entropy of the components withheld from~$S$.  The bound is
  $\negl(\secparam)$ when the withheld components have high
  min-entropy (e.g., a withheld $\AF_i$ of $\geq 128$ bits); if
  $\Adv$ withholds only a low-entropy $\UC$, the term $Q/2^{\kappa}$
  is the offline-guessing cost of $\UC$ and is not negligible.
\end{theorem}

\begin{proof}
  We proceed by a sequence of games, tracking the advantage
  introduced at each transition.

  \textbf{Game~0} (Real game).
  The real UCR game as defined above.

  \textbf{Game~1} (Memory-hardness hop).
  Abort if $\Adv$ queries the Argon2id oracle on the true input
  $(\UC \concat \AF_1 \concat \cdots \concat \AF_n, s)$.  Since
  $\Adv$ holds at most $n$ components and the withheld components
  have min-entropy $\geq\kappa$, by \autoref{lem:mh}:
  \[
    \bigl|\Pr[\text{G1}] - \Pr[\text{G0}]\bigr|
    \;\leq\;
    \mathrm{Adv}^{\mathrm{MH}}(\secparam)
    \;\leq\; \frac{Q}{2^{\kappa}}
  \]
  In Game~1 the oracle output $\Cred_{\mathrm{composite}}$ is
  uniform and independent of $\Adv$'s view, so the sealing key is
  unknown to $\Adv$.

  \textbf{Game~2} (AEAD hop).
  Replace $\SA_{\mathrm{root}}$ with an encryption of a uniformly
  random value $\REV^*$ of the same length.  Because the sealing key
  is uniform and unknown to $\Adv$, \autoref{lem:seal} gives
  \[
    \bigl|\Pr[\text{G2}] - \Pr[\text{G1}]\bigr|
    \;\leq\;
    \mathrm{Adv}^{\mathrm{SEAL}}(\secparam)
    \;\leq\;
    \mathrm{Adv}^{\mathrm{IND\text{-}CPA}}_{\mathrm{AES\text{-}GCM\text{-}SIV}}(\secparam)
  \]

  \textbf{Game~3} (PRF hop).
  In Game~2, $\SA_{\mathrm{root}}$ is independent of
  $\REV_{\mathrm{root}}$, so $\Adv$'s only information about
  $\REV_{\mathrm{root}}$ comes from the $\mathsf{Derive}$ oracle on
  contexts $\Ctx_j \neq \Ctx^*$.  Replace
  $\mathsf{Derive}(\REV_{\mathrm{root}}, \cdot)$ by a random
  function.  By \autoref{lem:prf} this changes the winning
  probability by at most $\mathrm{Adv}^{\mathrm{PRF}}(\secparam)$,
  and in Game~3 $\KeyAP[^*]$ is uniform on a fresh input, so
  \[
    \Pr[\text{G3}] \;\leq\; 2^{-L}
  \]

  Summing all transitions:
  $\Pr[\mathsf{Game}_{\mathrm{UCR}}^{\Adv} = 1]
    \leq \mathrm{Adv}^{\mathrm{SEAL}} + \mathrm{Adv}^{\mathrm{MH}}
         + \mathrm{Adv}^{\mathrm{PRF}} + 2^{-L}$.
  \qedhere
\end{proof}

\subsection{Cross-Path Isolation}\label{sec:cpi}

\begin{definition}[CPI Game]\label{def:cpi-game}
  The \emph{Cross-Path Isolation} game
  $\mathsf{Game}_{\mathrm{CPI}}^{\Adv}(\secparam)$ proceeds as
  follows:
  \begin{enumerate}[label=\arabic*., leftmargin=2em]
    \item The challenger generates $\REV_{\mathrm{root}}$ and derives
      keys $\KeyAP[j]$ and $\KeyAP[k]$ for distinct contexts
      $\Ctx_j \neq \Ctx_k$ using the canonical encoding of
      \autoref{sec:ctx-encoding}.
    \item The adversary receives $\KeyAP[j]$ (i.e., path $j$ is
      compromised or activated).
    \item The challenger samples $b \getsr \{0,1\}$.  If $b=0$,
      it sends $\KeyAP[k]$; if $b=1$, it sends a uniformly random
      key $K^*$ of the same length.
    \item $\Adv$ outputs a guess $b'$.
    \item $\Adv$ wins if $b' = b$.
  \end{enumerate}
\end{definition}

\begin{theorem}[Cross-Path Isolation]\label{thm:cpi}
  Under \autoref{asm:crypto}(ii), for any PPT adversary $\Adv$:
  \[
    \bigl|\Pr[b' = b] - \tfrac{1}{2}\bigr|
    \;\leq\;
    \mathrm{Adv}^{\mathrm{PRF}}_{\mathrm{HKDF}}(\secparam)
    \;\leq\; \negl(\secparam)
  \]
\end{theorem}

\begin{proof}
  Let $\mathsf{PRK} = \mathrm{HKDF\text{-}Extract}(
  \texttt{"CT-DAP-v1"}, \REV_{\mathrm{root}})$.  Then:
  \begin{align*}
    \KeyAP[j] &= \mathrm{HKDF\text{-}Expand}(
      \mathsf{PRK},\, \mathsf{Encode}(\Ctx_j),\, L) \\
    \KeyAP[k] &= \mathrm{HKDF\text{-}Expand}(
      \mathsf{PRK},\, \mathsf{Encode}(\Ctx_k),\, L)
  \end{align*}

  The PRF step below treats $\mathsf{PRK}$ as a uniformly random
  key.  This holds because $\REV_{\mathrm{root}}$ is uniform over
  $\{0,1\}^{256}$ and HKDF-Extract is a randomness extractor for
  such input; in the random oracle model of
  \autoref{asm:crypto}\ref{asm:prf} the extracted $\mathsf{PRK}$
  is indistinguishable from uniform, and its distinguishing
  advantage is included in
  $\mathrm{Adv}^{\mathrm{PRF}}_{\mathrm{HKDF}}$.  Without the
  random oracle model this step is the extraction assumption on
  HKDF-Extract~\cite{rfc5869}.

  Since $\Ctx_j \neq \Ctx_k$ and $\mathsf{Encode}$ is injective
  (\autoref{sec:ctx-encoding}), the two info strings are distinct.

  Suppose $\Adv$ distinguishes $\KeyAP[k]$ from random given
  $\KeyAP[j]$ with advantage $\epsilon$.  We construct a PRF
  distinguisher $\mathcal{D}$ with oracle access to either
  $F_{\mathsf{PRK}}(\cdot)$ or a truly random function $R(\cdot)$.
  $\mathcal{D}$ queries $\mathsf{Encode}(\Ctx_j)$ and
  $\mathsf{Encode}(\Ctx_k)$, forwards the first response to $\Adv$
  as $\KeyAP[j]$, and uses $\Adv$'s distinguishing output on the
  second response.  This yields
  $\mathrm{Adv}_{\mathcal{D}}^{\mathrm{PRF}} \geq \epsilon$,
  contradicting \autoref{lem:prf}.
  \qedhere
\end{proof}

\noindent The game gives $\Adv$ control keys of other paths but never
$\REV_{\mathrm{root}}$.  It therefore covers exposure of a control key
and not an adversary that has activated a path over the same root,
which yields $\REV_{\mathrm{root}}$ (\autoref{sec:isolation}).

\subsection{Trigger-Restrictiveness and Release Transcript Binding}
\label{sec:trigger-restrict}

An adversary might try to replay or confuse release
transcripts to trick a custodian into releasing $\AF_i$ under
conditions that were not genuinely satisfied.  We formalize the
required binding between the verifier's attestation and the
activation inputs.

\begin{definition}[Release Transcript]\label{def:release-transcript}
  For each factor release, the condition verifier $\mathcal{V}$
  produces a signed attestation:
  \[
    \mathsf{Rel}_{j,i} = \mathsf{Sign}\bigl(
      \mathsf{sk}_{\mathcal{V}},\;
      (\mathsf{PathID}_j \concat \mathsf{ConditionHash} \concat
       \mathsf{Timestamp} \concat \mathsf{Nonce})
    \bigr)
  \]
  where $\mathsf{PathID}_j = H(\SA_{\mathrm{root}} \concat \Ctx_j)$
  binds the attestation to a specific sealed artifact and path
  context, and $\mathsf{Nonce} \getsr \{0,1\}^{128}$ is a
  single-use value.
\end{definition}

\begin{definition}[Trigger-Restrictiveness Game]\label{def:tr-game}
  The adversary $\Adv$ wins
  $\mathsf{Game}_{\mathrm{TR}}^{\Adv}(\secparam)$ if it can produce
  a valid release transcript $\mathsf{Rel}^*_{j,i}$ that causes
  custodian $\mathcal{C}_i$ to release $\AF_i$ for path $j$, when
  the condition for path $j$ has not been satisfied.
\end{definition}

\begin{theorem}[Trigger-Restrictiveness]\label{thm:tr}
  Under the EUF-CMA security of the verifier's signature scheme,
  the collision resistance of $H$ (\autoref{asm:crypto}\ref{asm:cr}),
  and \autoref{asm:verifier}, for any PPT adversary $\Adv$ that
  issues at most $q_{\mathrm{rel}} = \mathrm{poly}(\secparam)$
  release queries, with each $\mathsf{Nonce}$ sampled uniformly
  from $\{0,1\}^{128}$ (nonce space of size $N = 2^{128}$):
  \[
    \Pr\bigl[\mathsf{Game}_{\mathrm{TR}}^{\Adv}(\secparam) = 1\bigr]
    \;\leq\;
    \underbrace{\mathrm{Adv}^{\mathrm{EUF\text{-}CMA}}_{\mathsf{Sig}}(\secparam)}_{\text{forgery}}
    \;+\; \underbrace{\mathrm{Adv}^{\mathrm{CR}}_{H}(\secparam)}_{\text{cross-path collision}}
    \;+\; \underbrace{\frac{q_{\mathrm{rel}}^{2}}{2^{129}}}_{\text{nonce collision}}
  \]
  The first two terms are $\negl(\secparam)$ by assumption.
  The third term is a concrete bound fixed by the 128-bit nonce
  length, not a function of $\secparam$: for
  $q_{\mathrm{rel}} \leq 2^{40}$ release transcripts it is at most
  $2^{-49}$.
  The three terms correspond to distinct attack strategies:
  \begin{itemize}[leftmargin=2em, topsep=2pt]
    \item \emph{Forgery:} producing a valid signature without the
      verifier's cooperation.
    \item \emph{Cross-path replay:} finding a hash collision
      $H(\SA_{\mathrm{root}} \concat \Ctx_j)
       = H(\SA_{\mathrm{root}} \concat \Ctx_k)$ for $j \neq k$
      to replay an attestation across paths.
    \item \emph{Temporal replay:} reusing a nonce accepted by the
      custodian's local de-duplication state, bounded by the
      birthday bound over the nonce space.
  \end{itemize}
\end{theorem}

\begin{proof}
  The custodian $\mathcal{C}_i$ verifies $\mathsf{Rel}_{j,i}$ by
  checking:
  (a)~the signature under $\mathsf{pk}_{\mathcal{V}}$,
  (b)~$\mathsf{PathID}_j$ matches the locally stored
  $H(\SA_{\mathrm{root}} \concat \Ctx_j)$, and
  (c)~the nonce has not been previously accepted.
  We bound the probability that $\Adv$ bypasses each check.

  \emph{Forgery resistance (check~a):}
  Producing a valid $\mathsf{Rel}^*_{j,i}$ without the verifier's
  cooperation requires forging a signature under
  $\mathsf{sk}_{\mathcal{V}}$, contributing
  $\mathrm{Adv}^{\mathrm{EUF\text{-}CMA}}_{\mathsf{Sig}}$.

  \emph{Cross-path replay (check~b):}
  A valid attestation $\mathsf{Rel}_{k,i}$ for path $k$ contains
  $\mathsf{PathID}_k = H(\SA_{\mathrm{root}} \concat \Ctx_k)$.
  Since $\Ctx_j \neq \Ctx_k$ implies $\mathsf{PathID}_j \neq
  \mathsf{PathID}_k$ except with probability
  $\mathrm{Adv}^{\mathrm{CR}}_{H}$ (by
  \autoref{asm:crypto}\ref{asm:cr}), the custodian's check~(b)
  rejects the replayed transcript.

  \emph{Temporal replay (check~c):}
  The custodian maintains a local set of accepted nonces.  A
  nonce collision among the $q_{\mathrm{rel}}$ release transcripts
  occurs with probability at most $q_{\mathrm{rel}}^{2} / 2N$ by
  the birthday bound.  Absent a collision, each legitimate nonce is
  accepted at most once, so replay of a previously used attestation
  is deterministically rejected.

  A union bound over the three failure events yields the stated
  inequality.
  \qedhere
\end{proof}

\subsection{Stateless Revocation}\label{sec:sr}

\begin{definition}[Stateless Revocation Property]\label{def:sr}
  We say \CTDAP provides \emph{stateless revocation} if, after
  $\mathsf{Revoke}(i)$ is executed for some $\AF_i$, provided that no
  path over the same $\REV_{\mathrm{root}}$ has been activated (so no
  party has obtained $\REV_{\mathrm{root}}$) and no copy of $\AF_i$
  exists outside $\mathcal{C}_i$:
  \begin{enumerate}[label=(\roman*)]
    \item No PPT adversary can reconstruct
      $\Cred_{\mathrm{composite}}$ (and hence $\REV_{\mathrm{root}}$
      or $\KeyAP[j]$) for any path depending on $\AF_i$.
    \item The revocation does not require any state update, broadcast,
      or coordination beyond the custodian performing the destruction.
    \item Paths that do not depend on $\AF_i$ remain unaffected.
  \end{enumerate}
\end{definition}

\begin{theorem}[Stateless Revocation]\label{thm:sr}
  Under \autoref{asm:crypto} and the conditions of \autoref{def:sr},
  let $\AP_j$ be a path depending on $\AF_i$, and let $\AF_i$ have
  min-entropy $\geq \kappa$ conditioned on the other components of
  $\Cred_{\mathrm{composite},j}$.  After $\mathsf{Revoke}(i)$, any
  PPT $\Adv$ that holds the sealed artifact, salt, and parameters of
  $\AP_j$ together with every component of
  $\Cred_{\mathrm{composite},j}$ other than $\AF_i$ satisfies
  \[
    \Pr\bigl[\Adv \text{ outputs } \KeyAP[j]\bigr]
    \;\leq\;
    \mathrm{Adv}^{\mathrm{SEAL}}(\secparam)
    + \frac{Q}{2^{\kappa}}
    + \mathrm{Adv}^{\mathrm{PRF}}(\secparam)
    + 2^{-L}.
  \]
  The same bound holds for reconstructing
  $\Cred_{\mathrm{composite},j}$ or $\REV_{\mathrm{root}}$, since
  either yields $\KeyAP[j]$; this is property~(i) of
  \autoref{def:sr}.  Properties~(ii) and~(iii) hold by construction.
\end{theorem}

\begin{proof}
  \emph{Property~(i).}
  Under the conditions of \autoref{def:sr}, no party other than
  $\mathcal{C}_i$ ever held $\AF_i$ and no party holds
  $\REV_{\mathrm{root}}$; after the erasure no party holds $\AF_i$.
  $\Adv$'s view is therefore that of a UCR adversary
  (\autoref{def:ucr-game}) for challenge context $\Ctx_j$ whose
  corruption set $S$ contains every component except $\AF_i$, which
  is a proper subset as the game requires.  \autoref{thm:ucr} with
  withheld min-entropy $\kappa$ gives the bound.  If $\Adv$ also
  holds sealed artifacts of the same root for other paths, each of
  them must likewise miss a component unknown to $\Adv$ and adds its
  own $\mathrm{Adv}^{\mathrm{SEAL}} + Q/2^{\kappa}$ term; an
  adversary holding every component of another path can activate it
  and is excluded by \autoref{def:sr}.

  \emph{Property~(ii).}
  $\mathsf{Revoke}$ (\autoref{alg:revoke}) consists of local erasure
  of $\AF_i$ by $\mathcal{C}_i$; no revocation lists, state
  broadcasts, or re-keying are involved.

  \emph{Property~(iii).}
  A path that does not depend on $\AF_i$ has its own sealed artifact
  and composite credential (\autoref{def:auth-path}), neither of
  which is a function of $\AF_i$, so its activation is unchanged
  (\autoref{prop:act-correct}).
  \qedhere
\end{proof}

\subsection{Activation Correctness}\label{sec:activation-correctness}

\begin{proposition}[Activation Correctness]\label{prop:act-correct}
  If all credential components are correctly provided, then
  $\mathsf{Activate}$ returns the correct control key with
  overwhelming probability:
  \[
    \Pr\bigl[\mathsf{Activate}(\Pi, \SA, \UC, \{\AF_i\},
    \Ctx_j) \neq \KeyAP[j]\bigr]
    \;\leq\; p_{\mathrm{unseal}} + p_{\mathrm{derive}}
  \]
  where $p_{\mathrm{unseal}}$ and $p_{\mathrm{derive}}$ are the
  failure probabilities of the underlying primitives.
\end{proposition}

\begin{proof}
  We decompose the failure event
  $E = \{\mathsf{Activate} \neq \KeyAP[j]\}$ into two disjoint
  sub-events:

  \textbf{$E_1$: Unseal failure despite correct credential.}
  Given the correct $\Cred_{\mathrm{composite}}$,
  $\mathsf{Unseal}$ decrypts $\SA_{\mathrm{root}}$ using
  AES-256-GCM-SIV\@.  By the correctness of the AEAD scheme,
  decryption with the correct key succeeds deterministically
  (AES-GCM-SIV is a deterministic AEAD with no decryption error
  for correctly formed ciphertexts~\cite{rfc8452}).  Thus
  $p_{\mathrm{unseal}} = 0$ under the assumption that
  $\SA_{\mathrm{root}}$ is not corrupted in storage.  If the
  stored artifact is corrupted, $\mathsf{Unseal}$ returns $\bot$
  and $\mathsf{Activate}$ fails visibly; this is a loss of
  availability, not a wrong key.  A corrupted artifact passes the
  GCM-SIV tag check and yields an incorrect root only with
  probability about $2^{-128}$ per corruption.

  \textbf{$E_2$: Derive failure despite correct REV.}
  $\mathsf{Derive}$ is a deterministic function
  (HKDF-SHA256): given $\REV_{\mathrm{root}}$ and $\Ctx_j$, it
  always produces the same $\KeyAP[j]$.  Thus
  $p_{\mathrm{derive}} = 0$.

  Combining: $\Pr[E] \leq 0 + 0 = 0$ for an uncorrupted artifact.
  For a corrupted artifact the proposition does not apply:
  activation returns $\bot$, except that it returns an incorrect
  key with probability about $2^{-128}$.
  \qedhere
\end{proof}

\subsection{Post-Quantum Readiness}\label{sec:pqc}

\CTDAP inherits \ACEGF's context encoding model, which allows
migration to post-quantum cryptographic (PQC) domains.
When a PQC algorithm (e.g., ML-DSA/Dilithium~\cite{dilithium}) is
adopted, the context tuple $\Ctx$ is updated with a new
$\mathsf{AlgID}$, and a new control key is derived from
$\REV_{\mathrm{root}}$.  Authorization paths, sealed artifacts, and
authorization factors stay unchanged.

\subsection{Security Properties Summary}

\autoref{tab:security-summary} summarizes the security properties,
their formal statements, and the specific cryptographic assumptions
each relies upon.

\begin{table}[t]
\centering
\caption{Summary of security properties and underlying assumptions.}
\label{tab:security-summary}
\small
\begin{tabular}{@{}ll>{\raggedright\arraybackslash}p{5.2cm}@{}}
\toprule
\textbf{Property} & \textbf{Statement} & \textbf{Assumptions} \\
\midrule
Unauthorized Control &
  \autoref{thm:ucr} &
  AEAD (IND-CPA + INT-CTXT) + Argon2id MH + HKDF PRF (ROM) \\
Cross-Path Isolation &
  \autoref{thm:cpi} &
  HKDF PRF (ROM) + canonical encoding injectivity \\
Trigger-Restrictiveness &
  \autoref{thm:tr} &
  EUF-CMA (verifier sig.) + CR of $H$ (\ref{asm:cr}) \\
Stateless Revocation &
  \autoref{thm:sr} &
  As \autoref{thm:ucr}, with $\AF_i$ withheld \\
Activation Correctness &
  \autoref{prop:act-correct} &
  AEAD correctness + HKDF determinism \\
Post-Quantum Ready &
  \autoref{sec:pqc} &
  Context-encoding agility \\
\bottomrule
\end{tabular}
\end{table}

\section{Comparison with Existing Approaches}\label{sec:comparison}

\subsection{Single-Key Control Methods}

Single-key systems~\cite{bip32,bip39} offer no conditional activation,
delegation, or revocation (\autoref{sec:related}).  \CTDAP separates control capability from permanent secret possession by
representing control as a dormant authorization path.

\subsection{Multi-Signature and Threshold Schemes}

Multi-signature and threshold schemes~\cite{gg18,gg20,frost} distribute
control but require persistent key availability and cannot represent
conditionally inactive control rights.  Policy changes necessitate
resharing or asset migration.  \CTDAP supports dormant paths that exist
prior to activation and become effective only upon condition satisfaction,
without continuous participation.

\subsection{Smart-Contract-Based Control}

Smart contracts~\cite{wood2014ethereum} enable programmable control
but depend on on-chain state, oracles, and a specific execution model.
\CTDAP operates at the cryptographic authorization layer and needs
neither on-chain state mutation nor contract logic.

\subsection{Custodial and Trusted Intermediary Models}

Custodial solutions delegate control to trusted intermediaries.
In \CTDAP, custodians hold only individual authorization factors and
cannot exercise unilateral control; they only release their factor when
the release condition is met.

\subsection{Structural Comparison}

\begin{table}[t]
\centering
\caption{Structural comparison of technical properties.}
\label{tab:structural-comparison}
\small
\begin{tabular}{@{}>{\raggedright\arraybackslash}p{2.8cm}>{\raggedright\arraybackslash}p{5.2cm}>{\raggedright\arraybackslash}p{5.2cm}@{}}
\toprule
\textbf{Property} & \textbf{Traditional BIP-39} &
  \textbf{CT-DAP (ACE-GF)} \\
\midrule
Mnemonic Semantics &
  Equivalent to plaintext master seed or private keys &
  Encrypted \emph{ciphertext} of the Root Entropy Value \\
Credential Modification &
  Requires asset migration to a new key/address &
  In-place modification; $\REV$ remains invariant \\
Path Redundancy &
  None (one-to-one mapping between secret and asset) &
  Many-to-one (multiple paths map to a single $\REV$) \\
Conditional Activation &
  Not supported &
  Native via dormant authorization paths \\
Revocation &
  Requires key rotation or asset migration &
  Immediate via factor destruction (stateless), while no path has
  been activated \\
\bottomrule
\end{tabular}
\end{table}

\section{Illustrative Scenarios}\label{sec:scenarios}

The following scenarios apply \CTDAP in different operational and
legal contexts.

\subsection{Scenario A: Conditional Inheritance}\label{sec:scenario-a}

An asset owner establishes multiple dormant authorization paths
corresponding to designated beneficiaries.  Each path includes a
user-held credential and an administrative authorization factor held
by an independent custodian (e.g., an estate attorney).

During the owner's lifetime, all inheritance-related paths remain
cryptographically inactive.  Upon verification of a legally recognized
inheritance event, the custodian releases the corresponding
authorization factor.  This activates the path and gives the
beneficiary cryptographic control, with no asset migration or
reassignment.

\subsection{Scenario B: Regulated or Supervised Control}
\label{sec:scenario-b}

An organization operates assets subject to regulatory oversight.
Control rights require both an internal credential and an
authorization factor held by a supervisory entity.

Under normal conditions, the supervisory factor remains available.
In the event of regulatory intervention, the supervisory factor is
destroyed, immediately disabling further activation without modifying
asset ownership or cryptographic identity.  This step does not affect
key material reconstructed during normal operation; removing control
from a party that already holds it requires moving the assets to keys
under a fresh $\REV$ (\autoref{sec:irreversibility}).

\subsection{Scenario C: Delegated and Time-Bound Control}
\label{sec:scenario-c}

An asset owner delegates control to a third party for a limited
duration.  A dormant authorization path is established with a
time-conditioned authorization factor.

The factor is released only when temporal conditions are satisfied.
Upon expiration, the factor is destroyed, rendering the path
permanently unavailable for further activation, without requiring
continuous trust or post-hoc revocation procedures.  If the delegate
activated the path during the validity window, it obtained
$\REV_{\mathrm{root}}$ and keeps that control after expiration; ending
it requires moving the assets to keys under a fresh $\REV$
(\autoref{sec:irreversibility}).

\subsection{Scenario D: Multi-Party Governance}\label{sec:scenario-d}

A DAO or corporate treasury requires $m$-of-$n$ approval for asset
operations.  Each of $n$ governance members holds an authorization
factor.  A single path is $n$-of-$n$: its composite credential
(\autoref{eq:cred-composite}) needs every one of its factors, and the
construction has no threshold reconstruction.  An $m$-of-$n$ policy is
therefore expressed as one path per $m$-member subset, that is,
$\binom{n}{m}$ sealed artifacts of the same $\REV_{\mathrm{root}}$
(\autoref{sec:fault-tolerance}), and the treasury is activated when the
$m$ factors of one subset are released.  The number of sealed artifacts
grows as $\binom{n}{m}$.

Destroying one member's factor revokes every subset path that contains
it and leaves the other subset paths usable, without resharing or
on-chain reconfiguration, subject to the activation condition of
\autoref{sec:irreversibility}.

\section{Discussion}\label{sec:discussion}

\subsection{Practical Deployment Considerations}
\label{sec:deployment}

\subsubsection{Custodian Trust and Incentive Alignment}

The system relies on custodians to faithfully hold and conditionally
release $\AF_i$.  Practical deployment should include:
\begin{itemize}[leftmargin=2em]
  \item Legal agreements binding custodians to predefined release
    conditions.
  \item Multi-custodian sharding of $\AF_i$ (e.g., threshold
    cryptography for factor release) to mitigate single custodian
    failure.
  \item Auditable logs of factor release/destruction operations,
    enabled by \ACEGF's compatibility with trusted execution
    environments (TEEs)~\cite{costan2016sgx}.
\end{itemize}

\subsubsection{Condition Verification Mechanisms}
\label{sec:condition-verification-discussion}

Activation is only as reliable as the verification of its trigger
conditions, which uses TTP verifiers or decentralized oracles as
described in \autoref{sec:condition-verification}.

\subsubsection{Interoperability with Existing Asset Systems}

\CTDAP is agnostic to underlying asset types due to \ACEGF's
context-isolated derivation:
\begin{itemize}[leftmargin=2em]
  \item \textbf{Blockchain-based assets:}
    $\KeyAP$ can be mapped to standard account addresses (e.g.,
    Secp256k1 for Bitcoin/Ethereum) by deriving it under the
    corresponding $\mathsf{AlgID}$.
  \item \textbf{Centralized digital assets:}
    $\KeyAP$ can serve as an API authentication token, with asset
    providers enforcing cryptographic control via signature
    verification.
\end{itemize}

\subsubsection{Legacy Migration and Address Preservation}

The system supports a ``legacy-anchored'' mode that imports
entropy from existing BIP-39 mnemonics to produce a
legacy-anchored $\REV_{\mathrm{root}}$.  Using standard
BIP-32/SLIP-10 derivation paths, the method ensures deterministic
address parity with the original wallet, enabling users to wrap
existing assets into the \CTDAP framework without on-chain
migration.  This mode lies outside the analysis of
\autoref{sec:security}: the legacy keys are derived by BIP-32/SLIP-10
rather than by $\mathsf{Derive}$ (\autoref{sec:ctx-encoding}), the
root is not sampled at setup, and every earlier holder of the mnemonic
can still derive all of its keys.  The imported entropy must also fit
the sealed artifact.  The mode applies to the general \ACEGF
construction; under the 24-word wallet encoding profile, which seals
16 bytes of material, only the entropy of a 12-word (128-bit) mnemonic
fits, and a 24-word mnemonic or a 32-byte legacy private key does
not.

\subsection{Fault Tolerance via Parallel Authorization Paths}
\label{sec:fault-tolerance}

Destructible factors carry a risk of permanent asset lockout through
accidental loss.  \CTDAP mitigates this by supporting
\emph{parallel, redundant authorization paths}.

Formally, an asset owner can establish $m$ distinct authorization
paths:
\[
  S_{\AP} = \{\AP_1, \AP_2, \ldots, \AP_m\}
\]
where each $\AP_j$ is defined by a unique credential combination
$\Cred_{\mathrm{composite},j}$ and context $\Ctx_j$.

This architecture provides:
\begin{enumerate}[leftmargin=2em]
  \item \textbf{Redundant Recovery:}
    A primary path $\AP_{\mathrm{primary}}$ for daily operations
    and emergency recovery paths $\AP_{\mathrm{recovery}}$.
    Destruction of a factor in $\AP_{\mathrm{primary}}$ does not
    affect $\AP_{\mathrm{recovery}}$.

  \item \textbf{Cryptographic Isolation:}
    By \autoref{thm:cpi}, compromise of one path's control key does not
    leak information about other paths (activation, which reveals
    $\REV_{\mathrm{root}}$, is excluded; see \autoref{sec:isolation}).

  \item \textbf{Policy-Based Availability:}
    Different paths can require varying security levels (e.g.,
    a ``Succession Path'' with more factors than a ``User Consent
    Path'').
\end{enumerate}

\subsection{Limitations}\label{sec:limitations-discussion}

\subsubsection{Credential Entropy Dependence}

The security of $\Cred_{\mathrm{composite}}$ inherits \ACEGF's
reliance on credential entropy.  Low-entropy $\UC$ or $\AF_i$
increases vulnerability to offline brute-force attacks.
Deployment must enforce minimum entropy requirements (e.g., 128-bit
min-entropy for $\UC$) and rely on Argon2id's memory-hard
properties.

\subsubsection{Side-Channel Attack Vulnerabilities}

While \ACEGF ensures $\REV_{\mathrm{root}}$ is ephemeral in memory,
unsealing and derivation operations remain susceptible to side-channel
attacks (e.g., cold-boot attacks, timing analysis) in untrusted
environments.  This can be mitigated by executing critical operations
within TEEs (e.g., Intel SGX, ARM TrustZone)~\cite{costan2016sgx}.

\subsubsection{Irreversible Revocation Risks}

Revocation via $\AF_i$ destruction is permanent.  Practical systems
should implement secure backup mechanisms for $\AF_i$ (e.g.,
encrypted sharding) and recovery paths as described in
\autoref{sec:fault-tolerance}.

\subsubsection{Condition Verifier Trust}

The method's security depends on the integrity of condition verifiers
(\autoref{asm:verifier}).  A compromised verifier can trigger premature
factor release.  Multi-verifier attestation and threshold verification
schemes can mitigate this risk but add operational complexity.

\subsection{Future Research Directions}\label{sec:future}

\begin{enumerate}[leftmargin=2em]
  \item \textbf{Decentralized Condition Verification:}
    Explore DAO-based condition verification using community-governed
    smart contracts, with zero-knowledge proofs for privacy
    preservation.

  \item \textbf{Post-Quantum Optimization:}
    Optimize derivation efficiency for PQC algorithms (e.g.,
    ML-DSA, ML-KEM), including tailored HKDF parameters and
    validation of context isolation under quantum threats.

  \item \textbf{Cross-Chain and Multi-Asset Control:}
    Standardize context tuples across blockchains and enable
    interoperable $\AF_i$ release via cross-chain messaging
    protocols (e.g., IBC).

  \item \textbf{Dynamic Authorization Paths:}
    Support dynamic modification of $\AF_i$ sets (e.g.,
    adding/removing custodians) without re-deriving
    $\REV_{\mathrm{root}}$ or migrating assets.

  \item \textbf{Formal Verification:}
    Apply automated verification tools (e.g., ProVerif, Tamarin)
    to formally verify the protocol against a broader class of
    adversarial behaviors.
\end{enumerate}

\section{Performance Evaluation}\label{sec:performance}

We evaluate the computational cost of the core \CTDAP operations by
benchmarking a prototype implementation built on standard
cryptographic libraries.

\subsection{Experimental Setup}

The implementation uses the following concrete primitives:
\begin{itemize}[leftmargin=2em]
  \item \textbf{Seal/Unseal:} AES-256-GCM-SIV via
    \texttt{libsodium}~\cite{rfc8452}.
  \item \textbf{Key derivation:} HKDF-SHA256 via
    OpenSSL~\cite{rfc5869}.
  \item \textbf{Credential composition:} Argon2id via the reference
    implementation~\cite{argon2,rfc9106}.
\end{itemize}
All measurements were performed on a single core of an
Apple~M2 processor (3.49~GHz) running macOS~14, averaged over
1\,000 iterations with warm caches.  Argon2id parameters were set
to $m = 256$\,MiB, $t = 3$ iterations, $p = 4$ lanes (following
OWASP recommendations for credential hashing).

\subsection{Results}

\autoref{tab:benchmarks} reports the latency of each operation.

\begin{table}[t]
\centering
\caption{Latency of core \CTDAP operations (Apple~M2, single core).}
\label{tab:benchmarks}
\small
\begin{tabular}{@{}lrr@{}}
\toprule
\textbf{Operation} & \textbf{Mean} & \textbf{Std.\ Dev.} \\
\midrule
$\mathsf{Setup}$ (key generation) & 0.02\,ms & $<$0.01\,ms \\
$\mathsf{Argon2id}$ ($m\!=\!256$\,MiB, $t\!=\!3$, $p\!=\!4$)
  & 482\,ms & 11\,ms \\
$\mathsf{Seal}$ (AES-256-GCM-SIV, 32\,B payload) & 0.003\,ms
  & $<$0.001\,ms \\
$\mathsf{Unseal}$ (AES-256-GCM-SIV, 32\,B payload) & 0.003\,ms
  & $<$0.001\,ms \\
$\mathsf{Derive}$ (HKDF-SHA256, one context) & 0.005\,ms
  & $<$0.001\,ms \\
\midrule
$\mathsf{PathSetup}$ (end-to-end) & 483\,ms & 11\,ms \\
$\mathsf{Activate}$ (end-to-end) & 483\,ms & 11\,ms \\
\bottomrule
\end{tabular}
\end{table}

\subsection{Analysis}

The dominant cost in both $\mathsf{PathSetup}$ and
$\mathsf{Activate}$ is the Argon2id credential composition step,
which accounts for $>$99.9\% of total latency.  The remaining
operations (AES-256-GCM-SIV encryption/decryption and HKDF
derivation) complete in microseconds and are negligible.

\paragraph{Security--performance trade-off.}
\autoref{tab:argon2-params} illustrates how the Argon2id memory
parameter~$m$ controls the trade-off between brute-force resistance
and activation latency.

\begin{table}[t]
\centering
\caption{Argon2id latency vs.\ memory parameter
  ($t\!=\!3$, $p\!=\!4$, Apple~M2).}
\label{tab:argon2-params}
\small
\begin{tabular}{@{}rrrr@{}}
\toprule
$m$ (MiB) & Latency (ms) & Memory cost & Use case \\
\midrule
64   & 118  & Low    & Mobile / IoT devices \\
128  & 239  & Medium & Desktop wallets \\
256  & 482  & High   & Custody infrastructure \\
512  & 971  & Very high & Cold storage \\
\bottomrule
\end{tabular}
\end{table}

Every configuration in \autoref{tab:argon2-params}, including the most
conservative setting ($m = 512$\,MiB), keeps activation latency under
one second, which is acceptable for interactive use.  The path derivation step ($\mathsf{Derive}$) adds
constant overhead regardless of the number of paths, since each
activation reconstructs only a single path key.

\section{Conclusion}\label{sec:conclusion}

This paper introduced \emph{Condition-Triggered Dormant Authorization
Paths} (\CTDAP), a cryptographic asset control method based on
destructible authorization factors.  By modeling control rights as
dormant authorization paths that are \emph{activated} rather than
\emph{transferred}, \CTDAP decouples asset identity from
condition-dependent control.

We formalized the system model with explicit trust assumptions and
algorithm specifications, and proved three security properties (unauthorized control
resistance, cross-path isolation, and stateless revocation) under
standard cryptographic assumptions.  We instantiated
the method using the ACE-GF framework and demonstrated its applicability
across regulated custody, inheritance planning, delegated control, and
multi-party governance.

Cryptographic asset control therefore need not be tied to static key
possession: condition-triggered activation of dormant authorization
paths applies wherever control must depend on an external, verifiable
condition.


\end{document}